\ifdefined\pdfminorversion\pdfminorversion=7\fi
\documentclass[sigconf,nonacm]{acmart}
\usepackage{amsmath,booktabs}
\usepackage{algorithm,algpseudocode}
\usepackage{tikz}
\usetikzlibrary{arrows.meta}
\newcommand{\OPT}{\operatorname{OPT}}

\makeatletter
\renewcommand{\qed}{%
  \ifmmode \mathqed
  \else \leavevmode\unskip\nobreak\hfill\nobreak\hspace{.5em}\hbox{\qedsymbol}%
  \fi}
\makeatother

\AtBeginDocument{%
  }

\setcopyright{none}

\title{Truthful by Omission: Click Thinning for LLM-Native Advertising}

\author{Junhua Liu}
\affiliation{\institution{University of Southern California}\city{Los Angeles}\country{USA}}
\author{Zhicheng Liang}
\affiliation{\institution{The Chinese University of Hong Kong, Shenzhen}\city{Shenzhen}\country{China}}
\author{Zejun Zhang}
\affiliation{\institution{University of Southern California}\city{Los Angeles}\country{USA}}
\author{Langzhou He}
\affiliation{\institution{University of Illinois Chicago}\city{Chicago}\country{USA}}
\author{Feng Qian}
\affiliation{\institution{University of Southern California}\city{Los Angeles}\country{USA}}
\renewcommand{\shortauthors}{Liu et al.}

\hypersetup{pdfauthor={Junhua Liu, Zhicheng Liang, Zejun Zhang, Langzhou He, Feng Qian},pdfkeywords={Ad auctions, LLM-native advertising, cascade model, truthful mechanisms, click thinning, welfare approximation}}
\begin{document}

\begin{abstract}
Generative AI assistants are becoming an advertising medium, embedding sponsored content in their answers. This setting departs from classical position auctions in two ways: an ad's click-through rate depends on the passage that hosts it, and because users stop at their first click, earlier ads take clicks from later ones. This competition for clicks makes truthfulness delicate: a higher bid can reshape the allocation and put a competitor ahead of the bidder, reducing the bidder's clicks. Existing mechanisms either attain a constant fraction of the optimal welfare (advertisers' total click value) but are not known to be truthful, or are truthful but their welfare guarantee degrades as the number of passages grows. We achieve both even when the answer and its editorial order are fixed before bidding. Assuming each ad's click-through rate when reached is known and unaffected by omitting other ads, we give a computationally efficient mechanism that is truthful in expectation and attains at least half the optimal welfare under free reordering.

Our key idea, click thinning, omits ads at random, showing earlier ads less often so that each ad's overall click probability no longer depends on the order of the ads. Welfare maximization then reduces to a tractable matching problem, and solving it exactly with VCG-style payments makes truthful bidding optimal in expectation. Characterizing the click probabilities that omission can achieve sharpens the approximation factor from $2$ to $2-(1-\delta)^{K-1}$ when at most $K$ ads are shown and each click-through rate is at most $\delta$. This factor is the best possible for any order-preserving rule, even one that ignores incentives and computation, so truthfulness and computational efficiency incur no additional worst-case loss. 
\end{abstract}

\maketitle

\section{Introduction}

Generative AI assistants are becoming an advertising medium, embedding sponsored content in their answers~\cite{zelch2024user,zhao2025llm}. Suppose a user asks an assistant to compare running shoes. The answer has several passages, such as a product comparison and a note on delivery, and each can host an ad. A shoe brand's ad fits the comparison better than the delivery note, so the platform must decide which advertiser appears in which passage~\cite{anchor}. Moreover, readers scan the answer from top to bottom and stop at their first click~\cite{craswell2008}, so an early ad takes clicks from every ad below it. LLM-native advertising thus extends position auctions~\cite{edelman2007} in two ways: an ad's click-through rate depends on its passage, and its overall click probability, which we call its \emph{click service}, depends on which ads precede it.

These interactions make truthfulness delicate. A truthful auction requires \emph{monotonicity} in click service, not merely in the chance of being shown: raising a bid must never lower the bidder's click service. In a bid-ranked position auction, a higher bid never moves an ad down. Here, it can change the whole matching and place a competitor ahead of the bidder. Balseiro et al.~\cite{anchor} formalize this problem: given public click-through rates for advertiser--passage pairs and an ad cap, the platform chooses a matching and a rendering order. For order-insensitive multinomial logit (MNL) choice, they compute optimal allocations exactly via linear-fractional matching. For the first-click cascade, they obtain either a $4/(1-\varepsilon)$ welfare approximation without a monotonicity guarantee or a randomized monotone $O(\log m)$-approximation for $m$ passages, and they leave open whether a constant approximation can also be monotone.

We study this question when the answer and the order of its passages are fixed before bidding. We call this order the \emph{editorial order}. In~\cite{anchor}, the LLM generates the answer in an order chosen from the bids. Fixing the order separates the content from bidding, as in designs that generate the organic answer before inserting ads~\cite{xu2026insertion}. The platform still chooses the matching and how often to show each ad, but it loses the lever that both cascade algorithms of~\cite{anchor} use to control competition for clicks: the constant-factor algorithm renders ads in decreasing order of value, and the monotone algorithm renders them in greedy selection order to secure monotonicity. Under editorial order, a high-value ad may sit behind competitors that absorb its clicks. Even the matching selected by our mechanism, if rendered in full, can let a higher bid bring a competitor ahead of the bidder and lower the bidder's click service (Proposition~\ref{prop:rawnonmonotone}).

Formally, advertisers have private nonnegative values per click, and each allowed pair $e$ of an advertiser and a passage, or \emph{context}, has a public \emph{standalone} click probability $p_e$: its click-through rate when a reader reaches it. Click trials are independent, at most $K$ ads may be shown, and omitting an ad leaves the answer and the other standalone probabilities unchanged (\emph{omission invariance}). This assumption is natural when ads are self-contained insertions that the renderer toggles without altering the surrounding text, and randomized omission makes it testable (Section~\ref{sec:model}). As in~\cite{anchor}, our benchmark $\mathrm{OPT}(v)$ is the maximum expected click value over feasible matchings and \emph{all} rendering orders.

We show that monotonicity and a constant approximation are compatible even under editorial order. Our computationally efficient mechanism is dominant-strategy truthful in expectation and attains at least $\mathrm{OPT}(v)/2$ at truthful reports (Theorem~\ref{thm:main}, Proposition~\ref{prop:monotone}). It charges only on clicks and keeps truthful utility nonnegative in every realization. Its outcomes are feasible in the model of~\cite{anchor}, so this answers their question, and the factor two also improves on their $4/(1-\varepsilon)$ guarantee, which needs reordering.

Our key idea, \emph{click thinning}, counters competition for clicks by deliberate omission. An early ad is reached by more readers, so we show it less often: before the user interacts, the platform keeps each selected ad independently with a probability computed from the standalone probabilities, smaller for earlier ads. With the right probabilities, every ad $e$ in a matching $M$ receives click service $p_e/(1+\sum_{f\in M}p_f)$, wherever it appears. For two ads with $p_e=1/2$, keeping the first with probability $1/2$ and the second with probability $2/3$ gives each a click probability of $1/4$ (Figure~\ref{fig:identity}). This order-invariant formula turns welfare into a ratio of linear functions of the matching, which can be maximized exactly in polynomial time. Because each matching's omission lottery is fixed by public data before bidding, exact maximization with pivot payments~\cite{clarke1971,groves1973,midr} charged only on clicks~\cite{lavi} is truthful in expectation.

\subsection{Contributions and Techniques}

\textbf{Click thinning (Section~\ref{sec:mechanism}).} Lemma~\ref{lem:service} derives the retention probabilities and the order-invariant click service, and every realized display is a feasible subset of the matching in editorial order. We maximize reported welfare exactly over this bid-independent range by linear-fractional matching, which keeps each advertiser's click service nondecreasing in its bid, even when its context, competitors, and retention probability all change. A threshold comparison~\cite{limits} bounds $\mathrm{OPT}(v)$ by twice the implemented welfare.

\textbf{Which click services omission can realize (Section~\ref{sec:odds}).} For a fixed matching and order, independent retention realizes every click-probability vector attainable by correlated omission (Lemma~\ref{lem:target}). For matchings with all $p_e<1$, the targets attainable in every order are given by per-ad weights bounded by the odds $p_e/(1-p_e)$ (Theorem~\ref{thm:invariant}). Optimizing over sub-matchings lets us use the odds weights without welfare loss, and thinning with these odds reproduces exactly the MNL click probabilities of~\cite{anchor}. Our odds mechanism combines these odds lotteries with unthinned certain-click singletons and maximizes welfare over all mixtures of individually order-invariant matching targets, with the same incentive and computational guarantees (Theorem~\ref{thm:complete}). It is welfare-optimal for $K=1$, and when every standalone probability is at most $\delta<1$, it attains at least $\mathrm{OPT}(v)/\alpha(K,\delta)$ with $\alpha(K,\delta)=2-(1-\delta)^{K-1}$.

\textbf{The cost of editorial order (Section~\ref{sec:editorial}).} For every $K\ge1$ and $0<\delta<1$, no order-preserving allocation rule guarantees a better worst-case approximation factor than $\alpha(K,\delta)$, even without incentive or computational restrictions, and the factor is tight as a supremum on disjoint edges (Theorem~\ref{thm:sharp}). Truthfulness and computational efficiency thus add no worst-case loss beyond fixing the order. The hard instances place a rare, high-value ad behind competitors that absorb its clicks; without a bound on standalone probabilities, two is the optimal uniform factor (Proposition~\ref{prop:tight}). In synthetic experiments, mean welfare relative to $\mathrm{OPT}(v)$ ranges from about 0.89 to 0.96 with multiple ads, and most of the gap comes from the order-invariant range rather than editorial order.

\begin{figure*}[t]
\centering
\includegraphics[ width=\textwidth]{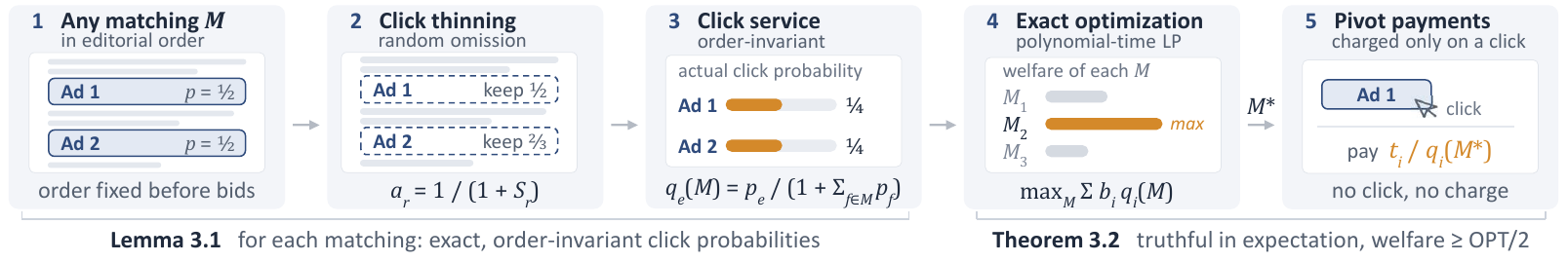}
 \vspace{-2em}
\caption{Structure of the basic mechanism. Click thinning realizes order-invariant click service. Exact range optimization with click-contingent pivot payments ensures truthfulness in expectation and welfare at least $\mathrm{OPT}(v)/2$ at truthful reports.}\label{fig:identity}
\end{figure*}

\subsection{Related Work}

\paragraph{LLM-native advertising.}
Balseiro et al.~\cite{anchor} solve MNL winner determination via the Charnes--Cooper transformation~\cite{charnes1962}, which we apply to thinned cascade service. Their cascade guarantees are Corollary~5.13 and Theorem~5.14. Xu et al.~\cite{xu2026insertion} also preserve the organic answer and its sentence order, using an additive genre-based advertiser--slot proxy with VCG and bounded-error incentives. In our model, service instead depends on competing ads through cascade reach. Other mechanisms use bid-conditioned response generation~\cite{zhao2025llm}, token or output distributions~\cite{dutting2024,soumalias2025}, LLM-generated summaries~\cite{dubey2024}, retrieved content~\cite{hajiaghayi2024}, quality-constrained allocations~\cite{han2026}, or endogenous ad content~\cite{wang2025content}.
Cascading bandits~\cite{kveton2015} learn attraction probabilities, which are public inputs here.

\paragraph{Cascade externalities.}
For ad-specific cascades, Kempe and Mahdian~\cite{kempe2008} and Aggarwal et al.~\cite{aggarwal2008} give exact algorithms and VCG mechanisms. Our free-order dynamic program (Appendix~\ref{app:dp}) extends theirs to context subsets. Farina and Gatti~\cite{farina} give a truthful public-order $1/2$ guarantee for context-independent quality and unit position-continuation factors. Their bound yields a subset of a free-order-optimal matching that keeps half its welfare in editorial order, and their analysis contains a related finite-product bound. Cloning advertisers across contexts to fit their model can select multiple copies of one advertiser,
violating feasibility (Appendix~\ref{app:boundary}). Our service range instead permits efficient truthful optimization over contextual matchings with the sharp cap--probability guarantee.

Gatti et al.~\cite[Proposition~5]{externalities} obtain a truthful $1/2$ guarantee in the ad--ad reset model by alternating occupied and empty slots. Empty contexts do not reset a first-click cascade, so this spacing construction does not transfer (Appendix~\ref{app:boundary}). Our retention probabilities instead compensate for differences in cascade reach.

\paragraph{Fractional optimization and truthful implementation.}
Gallego and Berbeglia~\cite{limits} connect fixed-preference independent consideration to cascades and develop auxiliary-MNL, threshold-excess, and attraction-monotonicity results. We apply their threshold comparison to feasible high-value subsets: thinning realizes the ratio as click welfare, and the no-click probability sharpens the bound. Gallego et al.~\cite{downward} study choice optimization under downward feasibility with totally unimodular constraints. Their last-choice bound supplies the relaxation behind our welfare certificates (Appendix~\ref{app:experiments}).

Computation uses integral matching polyhedra~\cite{schrijver2003}: our polynomial guarantee rests on a linear program, and our prototype follows Dinkelbach's ratio iteration~\cite{dinkelbach1967}. Exact optimization over public lotteries is maximal-in-distributional-range~\cite{midr}. Our contribution is the order-invariant range. Click-contingent pivots follow outcome-contingent payment realization~\cite{lavi}. We compute pivots explicitly, whereas implicit-payment transformations~\cite{bks} start from a monotone allocation. MacRury and Ma's inverse-safety attenuation~\cite{macrury} controls matching conflicts under online vertex arrivals, whereas thinning corrects for cascade reach within an already feasible matching (Appendix~\ref{app:boundary}). Adjacent models use position-only click-through rates in capacity-constrained auctions~\cite{capacity} or private ordinal preferences at fixed assortment prices~\cite{assortment}.

\section{Model}\label{sec:model}
We consider $n$ advertisers and $m$ contexts, each a passage of the answer with one insertion point. Let $E$ be the set of publicly allowed advertiser--context pairs, called edges, with public rational standalone probabilities $p_e\in[0,1]$. Feasible allocations are bipartite matchings $M\subseteq E$ with at most $K$ edges, for an integer ad cap $K\ge0$. Advertiser $i$ has a private value $v_i\ge0$ per click and reports a bid $b_i\ge0$. Rational inputs are assumed only for bit complexity. The incentive and approximation arguments hold for real values.

The editorial order of contexts is public and fixed before bidding. The platform renders each selected matching in the induced order and may omit sponsored creatives. Under a first-click cascade, a rendered ad on edge $e$ is clicked with probability $p_e$ conditional on no earlier click, and the user stops at the first click, which the platform observes for charging. Equivalently, with independent Bernoulli variables $Z_e$ of mean $p_e$, the user clicks the first rendered edge with $Z_e=1$. We assume that omitting other ads leaves the answer and $p_e$ unchanged (\emph{omission invariance}) and that an empty context causes no abandonment or reset. Both are natural when the renderer fixes the surrounding text, insertion positions, and styling before bidding and toggles only creatives. Standalone estimates do not establish this cascade model, but randomized omission can test it: systematic changes in retained-ad click rates conditional on reach, or exits before a click, would reveal violations.

\paragraph{Timing and information.}
The platform fixes contexts, allowed edges, probabilities, the cap, the editorial order, and the public tie rule before receiving reports. After the reports arrive, it computes the matching and payment parameters and then samples retention coins independently of Bernoulli click trials. It renders the retained ads in editorial order and charges the advertiser receiving the actual click, if any. The mechanism treats the public probabilities as correct inputs, and estimating them is a separate task.

An advertiser's utility is $v_i$ times the indicator of its actual click, minus its realized payment. \emph{Truthfulness in expectation} means that, for every fixed vector of others' reports, truthful reporting maximizes utility averaged over the platform's randomization and user clicks. \emph{Realized individual rationality} means that truthful utility is nonnegative in every realization.

The welfare objective is expected advertiser click value. We treat answer quality and total user--platform welfare as separate objectives. We compare our mechanism with the maximum expected welfare over the same feasible matchings when contexts may be freely reordered, denoted by $\OPT(v)$. This freely ordered optimum upper-bounds the best welfare under editorial order. Since welfare is linear in distributions over deterministic allocations, randomization cannot improve this optimum. For a fixed matching, decreasing-value order is optimal~\cite[Lemma~5.2]{anchor}: swapping adjacent values $u\ge v$ changes welfare by the nonnegative quantity $\text{reach}\cdot p_u p_v(u-v)$. Consequently, the freely ordered benchmark is
\begin{equation}\label{eq:benchmark}
\OPT(v)=\max_{M\text{ feasible}}\mathbb E\!\left[\max\bigl(\{v_iZ_e:e=(i,j)\in M\}\cup\{0\}\bigr)\right].
\end{equation}

\noindent\textit{Running example (Contextual auction).} Consider two advertisers, $A$ and $B$, and two contexts, $x$ and $y$. Both ads may be shown, so $K=2$. Their public probabilities and private values are
\[
\begin{array}{c|cc|c}
 &x&y&v_i\\\hline
A&4/5&1/5&4\\
B&3/5&9/10&3
\end{array}
\]
Contexts $x$ and $y$ can be read as two eligible passages of a generated answer. The two feasible two-ad matchings are the diagonal $\{Ax,By\}$ and the off-diagonal $\{Ay,Bx\}$. Selecting both edges incident to $x$, or both belonging to $A$, is infeasible: the auction must choose participants and contexts jointly. The freely ordered optimum renders $A$ before $B$, since $v_A>v_B$. The diagonal gives $4\cdot(4/5)+3\cdot(9/10)(1-4/5)=187/50$, whereas the reverse order gives $3\cdot(9/10)+4\cdot(4/5)(1-9/10)=151/50$. The off-diagonal in its better order gives $56/25$, and the best singleton gives $16/5$, so $\OPT(v)=187/50$. The diagonal's second advertiser has standalone probability $9/10$ but actual service only $9/50$ in the value-sorted rendering. Sections~\ref{sec:mechanism} and~\ref{sec:odds} revisit this instance to compare the mechanisms' click services and their welfare relative to $\OPT$.

\section{Click Thinning and the Basic Mechanism}\label{sec:mechanism}

We first show how randomized omission makes actual click welfare tractable to optimize without changing the editorial order. The basic construction gives click probabilities with a common denominator, so a single objective supports both truthful payments and a constant welfare guarantee. Figure~\ref{fig:identity} summarizes the construction.

\subsection{Exact Click Service}\label{sec:clickservice} Every matching inherits the editorial order, with public identifiers breaking any remaining ties. For $M=(e_1,\ldots,e_k)$, write $p_r=p_{e_r}$ and define
\begin{equation}\label{eq:retention}
S_r=\sum_{h=r}^{k}p_h,\quad a_r=\frac{1}{1+S_r}=\frac{1}{1+\sum_{h=r}^{k}p_h}.
\end{equation}
Put $S_{k+1}=0$. Retain $e_r$ independently with probability $a_r$ before the user interacts, and render the retained subsequence in editorial order. Let $D_M$ denote this lottery over rendered subsets, with user clicks governed by the cascade.

\begin{lemma}[Actual click service]\label{lem:service}
Under $D_M$, an edge $e\in M$ has actual click probability
\begin{equation}\label{eq:service}
q_e(M)=\frac{p_e}{1+S_1}=\frac{p_e}{1+\sum_{f\in M}p_f}.
\end{equation}
Every realization is a feasible subset of $M$, including when some probabilities are zero or one.
\end{lemma}
\begin{proof}
Write $q_r=q_{e_r}(M)$. Each prefix factor satisfies
\[
1-a_hp_h=1-\frac{p_h}{1+S_h}=\frac{1+S_{h+1}}{1+S_h}.
\]
Independence of retention coins and Bernoulli click trials gives $q_r=a_rp_r\prod_{h<r}(1-a_hp_h)$. Substitution gives
\[
q_r=\frac{p_r}{1+S_r}\prod_{h=1}^{r-1}\frac{1+S_{h+1}}{1+S_h}=\frac{p_r}{1+S_1},
\]
proving~\eqref{eq:service}. Retention probabilities lie in $[0,1]$ and no division by $p_r$ occurs. Every outcome is a subset of the matching and obeys the cap. The empty matching supplies zero service. 
\end{proof}

\noindent\textit{Example (Two-ad thinning).} Take two disjoint matched edges with $p_1=p_2=1/2$ in editorial order $(1,2)$. The retention probabilities are $a_1=1/2$ and $a_2=2/3$. The user reaches the second edge with probability $1-a_1p_1=3/4$, so both actual services equal $1/4$. Appendix~\ref{app:boundary} enumerates the rendered subsets.

\subsection{Truthfulness and Welfare Guarantees}
The service identity turns reported click welfare into a linear-fractional matching objective. For advertiser $i$, let $q_i(M)$ be its edge service if matched and zero otherwise. Using a fixed public lexicographic tie rule, choose $M^*(b)\in\operatorname*{arg\,max}_M R_b(M)$, where
\begin{equation}\label{eq:ratio}
R_b(M)=\sum_{i=1}^{n}b_iq_i(M)
=\frac{\sum_{e=(i,j)\in M}b_ip_e}{1+\sum_{e\in M}p_e}.
\end{equation}
Write $M^*=M^*(b)$ and $M^{-i}=M^*(0,b_{-i})$, the optimizer with $i$'s report set to zero. The pivot quantities are
\begin{equation}\label{eq:payment}
\begin{aligned}
H_i(b_{-i})&=\max_M\sum_{h\ne i}b_hq_h(M)=\sum_{h\ne i}b_hq_h(M^{-i}),\\
t_i&=\sum_{h\ne i}b_h\bigl[q_h(M^{-i})-q_h(M^*)\bigr].
\end{aligned}
\end{equation}
Both maximizations use the same bid-independent lotteries. Parameters are computed before drawing $D_{M^*}$. With outcome-contingent payments~\cite{lavi}, advertiser $i$ pays only on a click, at price
\begin{equation}\label{eq:charge}
c_i=\begin{cases}\displaystyle\frac{t_i}{q_i(M^*)},&q_i(M^*)>0,\\[2pt]0,&q_i(M^*)=0.\end{cases}
\end{equation}
Section~\ref{sec:computation} gives an exact polynomial-time algorithm for selection and payment computation.

Exact optimization gives the incentive guarantee, while a threshold comparison bounds the welfare of the implemented service.
\begin{theorem}[Basic mechanism]\label{thm:main}
In the model of Section~\ref{sec:model}, the basic mechanism of Section~\ref{sec:mechanism} is feasible, truthful in expectation, and individually rational in every realization at truthful reports. It never transfers money to an advertiser or charges without an actual click. At truthful reports its expected welfare is at least $\OPT(v)/2$. Matching selection and all payment parameters are computable in deterministic time polynomial in the rational input size. Exact fair-bit sampling takes expected polynomial time.\looseness=-1\par
\end{theorem}
\begin{proof}[Proof of Theorem~\ref{thm:main}: incentives and payments]
Each $D_M$ is fixed by public data. With the other reports held fixed, the expected utility of type $v_i$ under any report equals
\begin{align*}
v_iq_i(M^*)-t_i
=v_iq_i(M^*)+\sum_{h\ne i}b_hq_h(M^*)-H_i.
\end{align*}
Truthful reporting selects a maximizer of the first two terms, while $H_i$ is independent of $i$'s report. This proves dominant-strategy truthfulness in expectation, including at ties, once the stated expected payments are implemented.

\begingroup\predisplaypenalty=0
By definition $H_i$ is at least the others' welfare under $M^*$, so $t_i\ge0$. Nonnegative bids imply that optimal total reported welfare is at least $H_i$. Consequently
\begin{equation}\label{eq:bounds}
0\le t_i=H_i-\sum_{h\ne i}b_hq_h(M^*)\le b_iq_i(M^*).
\end{equation}
\endgroup
If $q_i=0$, then $t_i=0$. Otherwise $0\le c_i\le b_i$, and click-contingent charging has expectation $t_i$. At $b_i=v_i$, a click gives utility $v_i-c_i\ge0$ and no click gives zero. Feasibility follows from Lemma~\ref{lem:service}. Section~\ref{sec:computation} proves the computational and sampling guarantees.
\end{proof}
\begin{proof}[Proof of the welfare bound]
Put $r=\max_M R_v(M)$. If $r=0$, every feasible singleton has $v_ip_e=0$, so $\OPT=0$. Otherwise fix a matching $M_0$ attaining the maximum in~\eqref{eq:benchmark}. Its subset $T=\{e=(i,j)\in M_0:v_i>r\}$ is feasible. Since $R_v(T)\le r$,
\begin{equation}\label{eq:excess}
\sum_{e\in T}p_e(v_i-r)=\sum_{e\in T}v_ip_e-r\sum_{e\in T}p_e\le r.
\end{equation}
For every Bernoulli realization,
\begin{equation}\label{eq:threshold}
\begin{aligned}
\max\bigl(\{v_iZ_e:e\in M_0\}\cup\{0\}\bigr)
&\le r+\sum_{e\in M_0}(v_i-r)_+Z_e.
\end{aligned}
\end{equation}
Taking expectations and applying~\eqref{eq:benchmark} and~\eqref{eq:excess} gives $\OPT\le2r$. Lemma~\ref{lem:service} identifies $r$ with the mechanism's actual expected welfare. This applies the established threshold comparison~\cite{limits} to the implemented service vector.
\end{proof}

Downward feasibility lets the welfare proof discard edges below the threshold before invoking ratio optimality, since otherwise negative terms could cancel the positive excess. Independence underlies the service identity and the representation of $\OPT$ in~\eqref{eq:benchmark}, while the threshold inequality itself holds pointwise.

\begin{proposition}[Monotonicity of actual click service]\label{prop:monotone}
Fix the other advertisers' reports. Under the basic mechanism, each advertiser's actual click probability is nondecreasing in its report. The same holds for exact maximization over the range $\mathcal R^{\mathrm{odds}}$ of Section~\ref{sec:odds}.
\end{proposition}
\begin{proof}
Fix all reports except $i$'s, and compare reports $b_i<b_i'$. Let the two chosen lotteries have services $\boldsymbol{q}$ and $\boldsymbol{q}'$, and write $V_{-i}(\boldsymbol{q})=\sum_{h\ne i}b_hq_h$. Optimality at the two reports gives
\[
\begin{aligned}
b_iq_i+V_{-i}(\boldsymbol{q})&\ge b_iq'_i+V_{-i}(\boldsymbol{q}'),\\
b_i'q'_i+V_{-i}(\boldsymbol{q}')&\ge b_i'q_i+V_{-i}(\boldsymbol{q}).
\end{aligned}
\]
Adding yields $(b_i'-b_i)(q_i'-q_i)\ge0$: raising a report cannot decrease actual expected click service. This comparison lets the advertiser's context, competitors, and retention probabilities all change. Since it uses only exact maximization over a fixed range, it applies verbatim to $\mathcal R^{\mathrm{odds}}$.
\end{proof}

\noindent\textit{Running example (Basic mechanism).} It selects the diagonal, with click services $(8/27,1/3)$, click prices $(27/19,16/9)$, and welfare $59/27$, about $58\%$ of $\OPT(v)=187/50$. Appendix~\ref{app:running} gives the calculation and report-dependent allocations.

\subsection{Exact Computation}\label{sec:computation}
Selection and pivot payments require the same fractional matching optimization. We solve it for arbitrary finite public weights, allowing the click probabilities to be improved in Section~\ref{sec:odds} without changing the optimization algorithm.

The capped bipartite matching polytope $P$ consists of all $\boldsymbol{x}\ge0$ with advertiser row sums at most one, context column sums at most one, and $\sum_e x_e\le K$. It is integral, being the projection of an integral-capacity flow network with unit advertiser and context capacities and total flow at most $K$. Unused capacity is allowed.

\begin{lemma}[Exact linear-fractional matching]\label{lem:lfm}
For finite nonnegative rational weights $w$ and bids $b$, maximizing the ratio with numerator $\sum_{e=(i,j)}b_iw_ex_e$ and denominator $1+\sum_e w_ex_e$ over $P$ has the same optimal value as maximizing over capped matchings and as LP~\eqref{eq:lp}. Its value and lexicographically least maximizing matching are computable in time polynomial in the rational input size.
\end{lemma}
\begin{proof}
If $\boldsymbol{x}=\sum_M\lambda_M\mathbf1_M$ is a convex combination of matching incidence vectors, set $\mu_M=\lambda_M(1+\sum_{e\in M}w_e)/(1+\sum_{e\in E}w_ex_e)$. These coefficients are nonnegative and sum to one, and
\[
\frac{\sum_{e=(i,j)\in E}b_iw_ex_e}{1+\sum_{e\in E}w_ex_e}
=\sum_M\mu_M\frac{\sum_{e=(i,j)\in M}b_iw_e}{1+\sum_{e\in M}w_e}.
\]
The relaxation therefore preserves the optimal value: some integral matching attains the maximum over $P$.

Set $z=(1+\sum_e w_ex_e)^{-1}$ and $\boldsymbol{y}=z\boldsymbol{x}$. The Charnes--Cooper transformation, as used in~\cite{anchor}, yields the exact LP
\begin{equation}\label{eq:lp}
\begin{alignedat}{2}
\max_{\boldsymbol{y},z\ge0}\quad &\sum_{e=(i,j)\in E}b_iw_ey_e\\
\text{subject to}\quad
&\sum_{j:(i,j)\in E}y_{ij}\le z &\qquad &(i=1,\ldots,n),\\
&\sum_{i:(i,j)\in E}y_{ij}\le z &&(j=1,\ldots,m),\\
&\rlap{$\displaystyle\sum_{e\in E}y_e\le Kz,\quad z+\sum_{e\in E}w_ey_e=1.$}
\end{alignedat}
\end{equation}
For $K>0$ and nonempty $E$, $z\ge1/(1+K\max_e w_e)>0$. Zero-size instances are trivial. The change of variables is reversible. Exact rational LP algorithms compute its optimum $r$ with polynomial bit complexity.

The optimal ratio also lets us recover an integral optimizer. To recover a matching, maximize $\sum_{e=(i,j)}(b_i-r)w_ex_e$ over $P$. The definition of $r$ implies that this maximum is $r$, and every integral maximizer has ratio $r$. Standard capped weighted matching with successive exact edge-fixing queries selects the lexicographically least maximizing matching in polynomial time. Appendix~\ref{app:recovery} gives the exact recovery procedure.
\end{proof}

Setting $w=p$ gives the mechanism of Theorem~\ref{thm:main}. Only advertisers receiving positive service need separate pivot computations. If $q_i(M^*)=0$, the bound~\eqref{eq:bounds} forces $t_i=0$: $M^*$ already gives the others the full optimal reported welfare, so $H_i=R_b(M^*)$. Thus one allocation optimization and at most $K$ pivot-value optimizations suffice. Appendix~\ref{app:sampling} details encoding and exact fair-bit sampling.

\section{Characterizing and Optimizing Click Service}\label{sec:odds}
The basic mechanism uses one click-service target per matching. We characterize all attainable targets, then restrict to those implementable in every order. Their fractional representation enables exact polynomial-time optimization over contextual matchings. Optimizing over all fixed-order targets and matchings instead gives the fixed-order optimum, whose complexity for general contexts we leave open (Appendix~\ref{app:augmentation}).

\subsection{The Attainable Service Region}
A click can occur only if the user reaches the ad. We show that this reach constraint completely determines which click probabilities omission can realize. Fix an ordered matching $(e_1,\ldots,e_k)$ and write $p_r=p_{e_r}$. A preinteraction omission lottery may correlate its retention indicators but is independent of the Bernoulli click trials. Let $q_r$ denote the probability that the first click occurs on $e_r$.

\begin{lemma}[Fixed-order service region]\label{lem:target}
For a nonnegative vector $\boldsymbol{q}$, put $\rho_r=1-\sum_{h<r}q_h$. It is attainable by such a lottery if and only if
\begin{equation}\label{eq:region}
0\le q_r\le p_r\rho_r=p_r\left(1-\sum_{h=1}^{r-1}q_h\right),\quad r=1,\ldots,k.
\end{equation}
Whenever it is attainable, independent retention implements it. When $p_r\rho_r>0$, the unique independent retention probability for $e_r$ is $a_r=q_r/(p_r\rho_r)$. When $p_r\rho_r=0$, the target is zero and any retention probability is valid.
\end{lemma}
\begin{proof}
For arbitrary retention indicators $D$, independence from Bernoulli click trials gives
\[
q_r=p_r\Pr(D_r=1,\text{ no earlier click})\le p_r\rho_r.
\]
Earlier first-click events are disjoint, so their total probability is $1-\rho_r$. Conversely, inequalities~\eqref{eq:region} inductively keep every prefix sum in $[0,1]$. Set $a_r=q_r/(p_r\rho_r)$ wherever the denominator is positive and zero otherwise. Suppose the first $r-1$ targets have been realized, so that $e_r$ is reached with probability exactly $\rho_r$. The current retention coin and click trial are independent of these earlier events, giving actual service $a_rp_r\rho_r=q_r$. The next reach probability is therefore $\rho_{r+1}=\rho_r-q_r$, completing the induction from $\rho_1=1$. If reach becomes zero, all later targets are zero. The construction also covers total service one.
\end{proof}

For $q_r=p_r/(1+S_1)$, reach is $(1+S_r)/(1+S_1)$, recovering the suffix rule~\eqref{eq:retention}. The reach constraint also identifies targets that can be implemented regardless of the imposed editorial order.

\begin{theorem}[Targets attainable in every order]\label{thm:invariant}
For a fixed matching $M$, a nonnegative vector $\boldsymbol{q}$ with $q_0=1-\sum_{e\in M}q_e\ge0$ is attainable in every order, allowing a different preinteraction lottery for each order, exactly when the equivalent inequalities hold:
\begin{equation}\label{eq:invariant}
q_e\le p_e(q_0+q_e)\quad\Longleftrightarrow\quad(1-p_e)q_e\le p_eq_0,\qquad e\in M.
\end{equation}
For $q_0>0$, these targets are exactly
\[
q_0=\frac{1}{1+\sum_{f\in M}w_f},\qquad
q_e=w_eq_0=\frac{w_e}{1+\sum_{f\in M}w_f},
\]
where $w_e\ge0$ and $(1-p_e)w_e\le p_e$ for every $e\in M$.
For $q_0=0$, they are precisely the probability vectors supported on certain-click edges $p_e=1$.
\end{theorem}
\begin{proof}
Putting $e$ last in~\eqref{eq:region} gives $q_e\le p_e(q_0+q_e)$, which is~\eqref{eq:invariant}. Conversely, any prefix preceding $e$ has target mass at most $\sum_{f\ne e}q_f$, so~\eqref{eq:invariant} implies~\eqref{eq:region} in every order. If $q_0>0$, substitute $w_e=q_e/q_0$ and use $q_0+\sum q_e=1$. If $q_0=0$, the inequality forces $q_e=0$ whenever $p_e<1$, and Lemma~\ref{lem:target} implements the remaining vector.
\end{proof}

Retention adapts to each order while click probabilities remain fixed. The auction itself renders only in editorial order.

\subsection{Optimal Order-Invariant Service}
Theorem~\ref{thm:invariant} represents order-invariant targets by public weights bounded by the odds $o_e=p_e/(1-p_e)$ for $p_e<1$. Lemma~\ref{lem:weight-monotonicity} shows that optimizing over sub-matchings lets us raise these weights to the odds bounds without losing welfare, since downward closure permits removing edges below the optimal welfare ratio. Appendix~\ref{app:weights} also gives the finite-weight construction.

In the MNL model of~\cite{anchor}, an edge with $0<p_e<1$ has attraction $o_e$, the exponential of its log-odds. Attraction is zero when $p_e=0$. Thus, on any matching with all $p_e<1$, odds thinning realizes exactly the MNL click probabilities under the same $p_e$.

Certain-click edges have no finite odds, so we include their unthinned singletons. These allocations also cover targets with total click probability one. Let $E_{<}=\{e:p_e<1\}$ and $E_1=\{e:p_e=1\}$. Define the public range $\mathcal R^{\mathrm{odds}}$ to contain the empty outcome, every odds lottery on a feasible matching in $E_{<}$, and every unthinned feasible singleton in $E_1$. For $K\ge1$, the optimal reported welfare over this range is
\begin{equation}\label{eq:complete}
\begin{aligned}
r^{\mathrm{odds}}(b)&=\max\Bigl\{r_{<}(b),\max_{e=(i,j)\in E_1}b_i\Bigr\},\\
r_{<}(b)&=\max_{M\subseteq E_<}\frac{\sum_{e=(i,j)\in M}b_io_e}{1+\sum_{e\in M}o_e}.
\end{aligned}
\end{equation}
Here $M$ ranges over capped matchings, and a maximum over an empty set of singletons is zero. For $K=0$, the mechanism uses the empty outcome. Otherwise it selects a maximizing range entry by the public tie rule and computes all pivots over the same range. Algorithm~\ref{alg:auction} summarizes selection, payments, and retention.

\begin{algorithm}[t]
\caption{Odds mechanism}\label{alg:auction}
\begin{algorithmic}[1]
\Require Public $E,p,K$, editorial order and tie rule; bids $b\ge0$
\If{$K=0$} \State Render nothing, charge zero, and stop. \EndIf
\State Compute an exact optimizer of $r_{<}(b)$ in~\eqref{eq:complete}, allowing the empty matching.
\State Choose a welfare-maximizing $z^*$ among this optimizer and feasible $E_1$ singletons, using the public tie rule.
\State $q_i\leftarrow0$, $c_i\leftarrow0$ for every advertiser $i$.
\If{$z^*=\{(i,j)\}$ is a certain-click singleton}
\State $q_i\leftarrow1$; retain its edge with probability one.
\Else
\State Let $M=z^*=(e_1,\ldots,e_k)$ be the selected matching in editorial order, with $o_e=p_e/(1-p_e)$ and $O_r=\sum_{h=r}^k o_{e_h}$.
\State $q_i\leftarrow o_{(i,j)}/(1+O_1)$ for each $(i,j)\in M$.
\State $a_r\leftarrow o_{e_r}/[p_{e_r}(1+O_r)]$ if $p_{e_r}>0$, and $a_r\leftarrow0$ otherwise.
\EndIf
\For{each advertiser $i$ with $q_i>0$}
\State $H_i\leftarrow r^{\mathrm{odds}}(b_{-i},0)$ using both branches of~\eqref{eq:complete}.
\State $c_i\leftarrow\bigl[H_i-\sum_{h\ne i}b_hq_h\bigr]/q_i$.
\EndFor
\State If $z^*\subseteq E_{<}$, retain each $e_r$ independently with probability $a_r$, before user clicks.
\State Render retained ads in editorial order; charge $c_i$ only if $i$ clicks.
\end{algorithmic}
\end{algorithm}

\begin{theorem}[Optimal range of order-invariant matching targets]\label{thm:complete}
In the model of Section~\ref{sec:model}, $\mathcal R^{\mathrm{odds}}$ maximizes reported welfare over mixtures of targets that satisfy Theorem~\ref{thm:invariant} separately on each feasible matching. It dominates every finite implementable weight table and every finite menu of such tables. Exact selection with same-range pivots is truthful in expectation, individually rational in every realization at truthful reports, and achieves at least $\OPT(v)/2$. The computational and sampling guarantees of Theorem~\ref{thm:main} hold, with at most one allocation and $K$ pivot-value optimizations. For $K=1$, welfare is optimal.\looseness=-1\par
\end{theorem}
\begin{proof}
First consider targets with $q_0>0$ and no certain-click edges. Their weights satisfy $0\le w_e\le o_e$. Apply Lemma~\ref{lem:weight-monotonicity}, which requires only nonnegative weights, to the subsets of this matching, increasing $w$ to $o$. Some odds subset attains at least the target welfare.

For finite weights on a matching with uncertain part $U$ and certain part $C$, its ratio equals
\[
R_w(U\cup C)=\frac{\bigl(1+\sum_{e\in U}w_e\bigr)R_w(U)+\sum_{e=(i,j)\in C}b_iw_e}{1+\sum_{e\in U}w_e+\sum_{e\in C}w_e}.
\]
It is a convex combination of $R_w(U)$ and the certain-edge bids. Downward closure makes $U$ and every such singleton feasible. The range dominates each term. Targets with $q_0=0$ are convex combinations of certain-click singleton services, and mixtures cannot increase a linear welfare maximum. Every range entry is itself order-invariant, proving equality of the maxima.

The basic $w=p$ targets belong to this class, so dominance and Theorem~\ref{thm:main} give the half-welfare bound against the freely ordered optimum. Public range selection, pivots, and the charge bound prove incentives as before. The $r_{<}$ branch uses the exact fractional matching algorithm. The $E_1$ branch is a maximum over certain-edge bids, repeated in each positive-service pivot computation. Odds have polynomial rational encoding length. For $K=1$, the uncertain singleton has service $o_e/(1+o_e)=p_e$, and certain singletons have service one, giving the exact welfare optimum.
\end{proof}

\noindent\textit{Example (Aggregate order invariance).} The per-matching requirement in Theorem~\ref{thm:complete} matters. Requiring only aggregate advertiser service to be order-invariant defines a larger class. With two advertisers, two contexts, $K=2$, $p_{ij}=1/2$, and unit values, choose each perfect matching with probability $1/2$ and render it without omission. Each advertiser is first or second with equal probability, so its click service is $\tfrac12(\tfrac12+\tfrac14)=3/8$ under either context order. Total welfare is $3/4$, whereas the best odds matching gives $2/3$. Each component matching in this lottery has order-dependent edge services. Theorem~\ref{thm:complete} optimizes mixtures whose component targets are individually order-invariant.

\noindent\textit{Running example (Odds lottery).} The odds mechanism selects $Ax$ alone, with service $4/5$, click price $27/8$, and welfare $16/5$, about $86\%$ of $\OPT(v)$, up from $58\%$ under the basic mechanism.

Refining the threshold comparison of Theorem~\ref{thm:main} for odds weights yields the sharp cap--probability guarantee of Theorem~\ref{thm:sharp}. Appendix~\ref{app:augmentation} adds public order-dependent lotteries while preserving these guarantees and incentives.

\section{The Sharp Cost of Editorial Order}\label{sec:editorial}
\subsection{Tight Dependence on Cap and Probability}
The odds mechanism is optimal when at most one ad is allowed. With several ads, its welfare guarantee depends on the ad cap and the largest click probability. We show that this guarantee is sharp among all rules preserving editorial order: even without incentive or computational restrictions, no such rule has a better worst-case factor against the freely ordered optimum.
 The lower-bound instances place a low-probability, high-value ad after competing ads, whose earlier clicks prevent it from receiving service.

\begin{theorem}[Sharp guarantee]\label{thm:sharp}
If $K\ge1$ and $p_e\le\delta<1$ for every edge, then
\begin{equation}\label{eq:sharp}
\OPT(v)\le\alpha(K,\delta)\,r^{\mathrm{odds}}(v),\qquad\alpha(K,\delta)=2-(1-\delta)^{K-1}.
\end{equation}
For every $0<\delta<1$ and integer $K\ge1$, no smaller worst-case factor against $\OPT$ is attainable by any allocation rule preserving editorial order, even without incentive or computational restrictions. The bound is tight as a supremum on disjoint edges.
\end{theorem}
\begin{proof}
Write $r=r^{\mathrm{odds}}(v)$. If $r=0$, singleton optimality gives $\OPT=0$. Fix a matching $M_0$ attaining the maximum in~\eqref{eq:benchmark} and let $T=\{e=(i,j)\in M_0:v_i>r\}$. Since $T$ is feasible, its odds ratio is at most $r$, giving $\sum_T o_e(v_i-r)\le r$. If $T$ has a positive-probability edge, let $t$ be the smallest such probability. Since $p_e=(1-p_e)o_e$, the excess satisfies $\sum_{e\in T}p_e(v_i-r)\le(1-t)r$.
The threshold term is needed only when at least one click trial succeeds. Indeed, put $X=\max(\{v_iZ_e:e\in M_0\}\cup\{0\})$ and let $B$ be the event that some $Z_e=1$. Pointwise, $X\le r\mathbf1_B+\sum_{e\in T}(v_i-r)Z_e$: both sides vanish off $B$, and on $B$ the largest successful value is bounded by the threshold plus the positive excess. Taking expectations, as in the threshold comparison of~\cite{limits}, gives
\[
\OPT(v)\le r\left[1-\prod_{e\in M_0}(1-p_e)\right]+\sum_{e\in T}p_e(v_i-r).
\]
Set $A=(1-\delta)^{K-1}$. Since $M_0$ contains an edge of probability $t$ and at most $K-1$ others, its no-success probability is at least $(1-t)A$. Hence
\[
\OPT(v)\le r\bigl[1+(1-t)(1-A)\bigr]\le(2-A)r.
\]
If all edges in $T$ have probability 0, the excess vanishes, so $\OPT\le r$.

For tightness take $K$ disjoint edges: one has $p_H=\epsilon$, $v_H=1/\epsilon$, where $0<\epsilon\le\delta$. The other $K-1$ have probability $\delta$ and value one. Every odds subset containing $H$ has ratio one, since $v_Ho_H=1+o_H$, while subsets omitting $H$ have ratio below one. The $H$-first cascade attains the freely ordered optimum
\[
\OPT(v)=1+(1-\epsilon)(1-A)
\xrightarrow[\epsilon\downarrow0]{}2-(1-\delta)^{K-1}.
\]
Fix the editorial order with all unit-value edges before $H$. Since $p_Hv_H=1$, rendering $H$ with $\ell$ unit-value edges yields welfare
\[
\bigl[1-(1-\delta)^\ell\bigr]+(1-\delta)^\ell p_Hv_H=1.
\]
Renderings without $H$ have welfare at most one. Thus the best fixed-order welfare is exactly one, and every lottery over such renderings has welfare at most one. This proves the lower bound for all order-preserving rules. For irrational $\delta$, rational probabilities approaching $\delta$ from below give the same supremum. For $K=1$ the factor is exactly one.
\end{proof}

Reversing the order of this family separates order loss from range loss. With $H$ first, the fixed-order optimum equals $\OPT=1+(1-\epsilon)[1-(1-\delta)^{K-1}]$, yet the odds welfare remains one, so the entire loss comes from the range. Writing $F$ for the fixed-order optimum, the bound~\eqref{eq:sharp} holds against $F$ since $F\le\OPT$, and this family shows that it is tight. The public maximum-$\sum p_e$ augmentation (Appendix~\ref{app:augmentation}) selects all $K$ edges, and their unthinned rendering attains the fixed-order optimum in both orientations: $\OPT$ with $H$ first and one with $H$ last.

\begin{proposition}[Optimal uniform factor]\label{prop:tight}
The optimal uniform worst-case factor against the freely ordered optimum among all allocation rules preserving editorial order is two.
\end{proposition}
\begin{proof}
Take two disjoint edges with $p_H=\epsilon$, $v_H=1/\epsilon$, $p_L=1-\epsilon^2$, and $v_L=1$. Put $L$ before $H$. The best fixed-order welfare is one, while the freely ordered welfare is $1+(1-\epsilon)(1-\epsilon^2)\to2$. Theorem~\ref{thm:complete} supplies the factor-two guarantee for all instances.
\end{proof}

\subsection{Thinning versus Direct Rendering}
The odds optimizer selects a matching as well as a click target. If value-based reordering is permitted, the entire matching can be rendered without omission while retaining the welfare guarantee.

\begin{corollary}[Deterministic free-order allocation]\label{cor:unthinned}
Select a welfare-maximizing entry of $\mathcal R^{\mathrm{odds}}$ at values $v$, render its underlying matching without omission, and order it by decreasing value. This deterministic allocation has expected welfare at least $\OPT(v)/2$. Under the hypotheses of Theorem~\ref{thm:sharp}, it attains the stronger factor $2-(1-\delta)^{K-1}$.
\end{corollary}
\begin{proof}
Fix the selected matching $M$ and couple all renderings using the same Bernoulli variables $Z_e$. The full value-ordered rendering obtains $\max(\{v_iZ_e:e\in M\}\cup\{0\})$. For every omission realization, the first click in any retained subsequence has value at most this maximum. Taking expectations shows that the full value-ordered welfare dominates the odds welfare. Apply Theorems~\ref{thm:complete} and~\ref{thm:sharp}.
\end{proof}

Direct rendering in editorial order need not preserve the odds optimizer's monotonicity: a higher bid can bring a competitor ahead of the bidder, reducing the bidder's click service.

\begin{proposition}[Direct rendering]\label{prop:rawnonmonotone}
Optimizing the odds ratio and rendering the selected matching without omission in editorial order need not yield monotone actual click service.
\end{proposition}
\begin{proof}
Take $K=2$ and contexts ordered $1<2<3$. The allowed edges are $B1$, $A2$, and $C3$. Their probabilities are $p_A=4/5$, $p_B=3/10$, and $p_C=9/10$. Fix $b_B=2$, $b_C=1$, and vary $b_A=x$. The pairwise odds ratios are $R(AC)=(4x+9)/14$, $R(AB)=(14x+3)/19$, and $R(BC)=69/73$. The singleton ratios are $4x/5$, $3/5$, and $9/10$. The first two pairwise ratios cross at $x=43/40$. At $x=107/100$, $AC$ uniquely maximizes the ratio over all feasible sets. At $x=108/100$, $AB$ uniquely maximizes it. In $AC$, $A$ is first and clicks with probability $4/5$. In $AB$, $B$ precedes $A$, reducing $A$'s service to $(1-3/10)4/5=14/25$. Thus increasing $A$'s bid strictly decreases its actual service. In contrast, odds thinning implements services $2/7$ and $14/19$, respectively.
\end{proof}

Adding the two incentive inequalities for a lower and a higher scalar type shows that any truthful payment rule requires nondecreasing expected service. Hence no payments make this direct-rendering allocation rule truthful in expectation.

\setlength{\textfloatsep}{12pt}
\section{Evaluation}\label{sec:evaluation}
Our synthetic comparisons address three questions: (1) how much welfare is lost to editorial order versus the mechanism's range, (2) how our mechanism compares with monotone bucket allocation across probability scales, and (3) when exact free-order optimization becomes computationally costly.

\subsection{Order Loss and Range Loss}\label{sec:rangecomparison}
We enumerate 144 four-advertiser, four-context instances with separable, fully contextual, and bottleneck structures, $K\in\{1,2,4\}$, $p_e\le\delta\in\{0.1,0.8\}$, and eight repetitions per combination. Values lie in $\{1/4,1,4,16\}$. Let $W$ be the welfare of the odds mechanism (odds for short), $F$ the exact fixed-order optimum, and $O=\OPT$ the freely ordered optimum. Exact fixed-order optimization with pivot payments truthfully implements $F$ (Appendix~\ref{app:augmentation}). For every positive-optimum instance, $W/O=(W/F)(F/O)$, where $W/F$ and $F/O$ measure welfare retained by the range and under editorial order, respectively. Appendix~\ref{app:experiments} specifies the generators and comparisons.

\begin{table}[t]
\centering\small
\captionsetup{position=bottom,skip=3pt}
\setlength{\tabcolsep}{2pt}
\setlength{\aboverulesep}{1pt}
\setlength{\belowrulesep}{1.5pt}
\renewcommand{\arraystretch}{0.96}
\begin{tabular*}{\columnwidth}{@{\extracolsep{\fill}}ccrrrrr@{}}
\toprule
\multicolumn{2}{c}{Setting} & \multicolumn{3}{c}{Welfare $W/O$} & \multicolumn{2}{c}{Decomposition}\\
\cmidrule(lr){1-2}\cmidrule(lr){3-5}\cmidrule(l){6-7}
$\delta$ & $K$ & Mean & Min. & Guarantee & $W/F$ & $F/O$\\
\midrule
0.1 & 2 & 0.963 & 0.930 & 0.909 & 0.970 & 0.992\\
    & 4 & 0.931 & 0.873 & 0.787 & 0.943 & 0.987\\
\addlinespace[1pt]
0.8 & 2 & 0.894 & 0.6427 & 0.556 & 0.908 & 0.985\\
    & 4 & 0.916 & 0.766 & 0.502 & 0.945 & 0.970\\
\bottomrule
\end{tabular*}
\caption{Odds welfare ratios over 24 instances per row. $W/F$ and $F/O$ are means. Guarantee is $1/\alpha(K,\delta)$ (Theorem~\ref{thm:sharp}).}\label{tab:ranges}
\end{table}

For $K=1$, odds, fixed-order OPT, and free-order OPT coincide. For multiple ads (Table~\ref{tab:ranges}), mean $F/O$ lies between $0.970$ and $0.992$, while mean $W/F$ lies between $0.908$ and $0.970$. The range restriction accounts for most of the mean loss in this sample. The minimum individual $W/O$ is $0.6427$. In expectation, odds omits $3.8\%$ and $12.8\%$ of its selected ads at $\delta=0.1$ and $0.8$, respectively (Appendix~\ref{app:structure}). Table~\ref{tab:methods} compares basic thinning, odds, and augmentation (Aug.) with the truthfully implementable best singleton (Single), the unthinned odds-selected matching (Raw), and matching with weights $v_ip_{ij}$ (Stand.). Raw retains more mean welfare than odds in every multiple-ad setting, but its actual click service need not be monotone in bids (Proposition~\ref{prop:rawnonmonotone}). Adding one public unthinned matching to the odds range improves welfare in 25 of the 96 multiple-ad instances while preserving the sharp guarantee.

\begin{table}[t]
\centering\small
\captionsetup{position=bottom,skip=3pt}
\setlength{\tabcolsep}{2pt}
\setlength{\aboverulesep}{1pt}
\setlength{\belowrulesep}{1.5pt}
\renewcommand{\arraystretch}{0.96}
\begin{tabular*}{\columnwidth}{@{\extracolsep{\fill}}ccrrrrrr@{}}
\toprule
\multicolumn{2}{c}{Setting} & \multicolumn{3}{c}{Baselines} & \multicolumn{3}{c}{Proposed mechanisms}\\
\cmidrule(lr){1-2}\cmidrule(lr){3-5}\cmidrule(l){6-8}
$\delta$ & $K$ & Single & Raw & Stand. & Basic & Odds & Aug.\\
\midrule
0.1 & 2 & 0.795 & 0.987 & 0.987 & 0.905 & 0.963 & 0.968\\
    & 4 & 0.721 & 0.974 & 0.974 & 0.876 & 0.931 & 0.962\\
\addlinespace[1pt]
0.8 & 2 & 0.830 & 0.952 & 0.977 & 0.6431 & 0.894 & 0.909\\
    & 4 & 0.874 & 0.948 & 0.642 & 0.626 & 0.916 & 0.929\\
\bottomrule
\end{tabular*}
\caption{Mean welfare/free-order OPT. All rules preserve editorial order. Raw and Stand. lack incentive guarantees.}\label{tab:methods}
\end{table}

\subsection{Comparison with Bucket Allocation}
On 48 instances with $K=m=4$, we compare odds with Balseiro et al.'s~\cite{anchor} monotone bucket allocation and a variant that samples only nonempty buckets. Bucket membership and emptiness depend only on public data, so both bid-independent mixtures preserve monotonicity. Both sample buckets uniformly, render in greedy selection order, and return empty output if all buckets are empty.

At $\delta=0.1$, mean welfare/free-order OPT is $0.247$ for the published allocation, $0.988$ after removing empty buckets, and $0.931$ for odds. All edges lie in bucket four, so removing empty buckets quadruples expected welfare. Both bucket variants allow reordering, while odds preserves editorial order. At $\delta=0.8$, the means are $0.403$, $0.441$, and $0.916$. Thus, in mean welfare, odds trails the nonempty-bucket variant at $\delta=0.1$ but is more than twice either baseline at $\delta=0.8$, without reordering.

\subsection{Computational Cost}\label{sec:runtime}
With few contexts, a subset dynamic program computes the freely ordered optimum. After sorting, it uses $O(nm\sum_{h=0}^{\min(K,m)}\binom mh)$ arithmetic operations and is polynomial for fixed $K$. Appendix~\ref{app:dp} gives the recurrence. On 36 paired synthetic instances, exact DP-VCG can be faster for small context sets and caps. At $m=12,K=4$, median allocation/pivot times are $747/2831$ ms for DP and $24/112$ ms for odds. These timings compare exact-arithmetic implementations: DP-VCG attains the freely ordered optimum, while odds preserves editorial order. Appendix~\ref{app:paired} gives the comparison and protocol.

The rational odds prototype maximizes $\sum_e(b_i-r)o_e$ by exact integer min-cost flow at each ratio $r$, reaching optimality when the residual is zero. The polynomial guarantee rests on LP~\eqref{eq:lp}, not the iteration count. Appendix~\ref{app:experiments} reports sparsity measurements and rational upper certificates when exact OPT is unavailable.

\section{Conclusion}
Editorial order need not be sacrificed for truthful and efficient ad allocation. Click thinning prescribes actual click probabilities by omitting ads before interaction, turning implemented click welfare into a linear-fractional matching objective optimized exactly over a bid-independent range. The resulting mechanism has monotone click service and attains at least half the freely ordered optimum, answering the constant-approximation monotonicity question of~\cite{anchor}. Characterizing attainable click service sharpens the approximation factor to $2-(1-\delta)^{K-1}$ for ad cap $K\ge1$ and probabilities bounded by $0<\delta<1$. This is the optimal worst-case factor against the freely ordered optimum among all order-preserving rules, with no additional approximation loss from truthfulness or computation. Appendix~\ref{sec:calibration} bounds the effects of probability errors within the independent, omission-invariant model. Learning click probabilities, maximizing revenue, and adapting the answer to selected ads remain directions for future work in this setting.

\bibliographystyle{ACM-Reference-Format}
\bibliography{reference}

\appendix
\setlength{\textfloatsep}{10pt}
\section{Exact Computation and Sampling}\label{app:recovery}\label{app:certificate}
For finite nonnegative rational weights $w$, set $\boldsymbol a=(b_iw_e)_{e=(i,j)\in E}$ and $\boldsymbol w=(w_e)_{e\in E}$, and let $P$ be the convex hull of feasible incidence vectors. As in Lemma~\ref{lem:lfm}, at the optimal ratio $r$, the linear objective $(\boldsymbol a-r\boldsymbol w)\cdot\boldsymbol x$ has maximum $r$ over $P$, and its integral maximizers are exactly the ratio-optimal allocations.

Visit coordinates in public order. Tentatively fix each to zero, retain this choice if the restricted linear optimum remains $r$, and otherwise fix it to one. Coordinate restrictions are faces of the $0$--$1$ polytope, so the surviving optimal face always contains an integral vertex. After at most $|E|$ queries, the remaining incidence vector is the lexicographically least optimum. The queries are constrained weighted-matching optimizations for capped matchings and exact LPs for the integral-polyhedral case of Corollary~\ref{cor:family}. Their number and rational encoding lengths are polynomial.

\noindent\textit{Running example (Optimizer certificate).} At $w=p$, bids $(4,3)$ and $r=59/27$, the coefficients $(b_i-r)p_{ij}$ for $(Ax,Ay,Bx,By)$ are $(196/135,49/135,22/45,11/15)$. The diagonal sums to $r$, the off-diagonal to $23/27$, and every singleton coefficient is at most $r$. With the empty matching at zero, this certifies exact optimality over all feasible matchings.

\subsection{Rational Encoding and Exact Sampling}\label{app:sampling}
Retention probabilities and charges have polynomial rational bit length. For rational inputs, both basic weights and finite odds have polynomial bit length. Each service and retention probability uses sums of at most $K$ such weights and a fixed number of rational products and quotients. A common denominator has bit length at most the sum of the constituent denominator lengths. Exact LP optimization and matching recovery likewise give polynomial-bit pivot values, so the expected payments $t_i$ have polynomial bit length. Dividing by a positive $q_i$ preserves this bound even when its numerical value is small, and charges satisfy $c_i\le b_i$.

To sample a rational coin $u/d$, draw $\lceil\log_2d\rceil$ fair bits, reject integers at least $d$, and retain the ad exactly when the accepted integer is below $u$. Endpoint coins are deterministic. Acceptance probability exceeds $1/2$, so expected bit use and time are polynomial and termination is almost sure. A deterministic bit bound is unavailable for arbitrary rational coins. An ideal rational-coin primitive instead uses $k$ draws. This cost excludes user clicks.

\section{Comparison with Related Constructions}\label{app:boundary}
\begin{table}[t]
\centering
\begin{tabular*}{\columnwidth}{@{\extracolsep{\fill}}lcccc@{}}
\toprule
Subset $S$ & $\Pr(S)$ & Click 1 & Click 2 & No click\\
\midrule
$\varnothing$ & $1/6$ & $0$ & $0$ & $1$\\
$\{1\}$ & $1/6$ & $1/2$ & $0$ & $1/2$\\
$\{2\}$ & $1/3$ & $0$ & $1/2$ & $1/2$\\
$\{1,2\}$ & $1/3$ & $1/2$ & $1/4$ & $1/4$\\
\midrule
Unconditional & $1$ & $1/4$ & $1/4$ & $1/2$\\
\bottomrule
\end{tabular*}
\caption{Two-ad thinning (Section~\ref{sec:clickservice}). Click and no-click probabilities are conditional on $S$ except in the last row. Unequal retention ($1/2$, $2/3$) yields equal click service.}\label{tab:physical}
\end{table}

In Table~\ref{tab:physical}, showing both ads halves the second ad's click probability to $1/4$, although the matching is feasible. Random-order contention resolution~\cite{macrury} handles conflicts under vertex-batched activation. Here a fixed capped matching already ensures feasibility, and thinning compensates for cascade reach.

Farina and Gatti's related finite-product bound occurs in the $z=K$ case~\cite[Lemma~32]{farina}. In their context-independent ad model, cloning $(i,x)$ and $(i,y)$ as distinct ads allows both to be selected, violating the advertiser constraint. Merging them loses their different click probabilities. A valid reduction must preserve contextual probabilities and one strategic advertiser per original bidder. The reset model~\cite{externalities} restores attention through empty slots. By contrast, two half-probability ads separated by an empty context still have total first-click probability $1/2+(1-1/2)/2=3/4$, rather than one.

\section{Mechanism Extensions}\label{app:extensions}
\subsection{Public Weights and Welfare Monotonicity}\label{app:weights}\label{app:nested}
\begin{theorem}[Finite public weights]\label{thm:weights}
Fix finite nonnegative rational weights satisfying
\begin{equation}\label{eq:weights}
p_e\le w_e,\qquad 0\le(1-p_e)w_e\le p_e,\qquad e\in E.
\end{equation}
For each matching in editorial order, write $p_\ell=p_{e_\ell}$, $w_\ell=w_{e_\ell}$ and $W_\ell=\sum_{h\ge\ell}w_h$. Retain $e_\ell$ with probability $a_\ell=w_\ell/[p_\ell(1+W_\ell)]$ for $p_\ell>0$ and zero otherwise. Replace $p$ by $w$ in the ratio objective and payments. All conclusions of Theorem~\ref{thm:main} continue to hold.
\end{theorem}
\begin{proof}
If $p_e=0$, the inequalities force $w_e=0$. For positive $p_\ell$, the second inequality gives $w_\ell\le p_\ell(1+W_\ell)$, so retention is valid. The prefix factors telescope as $(1+W_{h+1})/(1+W_h)$, including zero-weight edges, yielding $q_e(M)=w_e/(1+\sum_{f\in M}w_f)$ for $e\in M$. The exact computation and incentive arguments are unchanged. For the comparator subset $T$ defined by $v_i>r$, ratio optimality gives $\sum_Tw_e(v_i-r)\le r$. Since $p\le w$, the expected threshold excess in~\eqref{eq:threshold} is no larger. Hence $\OPT\le2r$.
\end{proof}

\begin{lemma}[Monotonicity in public weights]\label{lem:weight-monotonicity}
For nonnegative bids and a nonempty downward-closed family $\mathcal F\subseteq2^E$ on finite $E$, the maximum of $R_{b,w}(M)=\sum_{e=(i,j)\in M}b_iw_e/(1+\sum_{e\in M}w_e)$ over $M\in\mathcal F$ is coordinatewise nondecreasing in finite nonnegative weights $w$.
\end{lemma}
\begin{proof}
Remove zero-weight edges from an optimum $M$ of ratio $r$. Set $D=1+\sum_{e\in M}w_e$. Downward closure and optimality give
\[
R_{b,w}(M\setminus\{e\})-r
=\frac{rD-b_iw_e}{D-w_e}-r
=\frac{w_e(r-b_i)}{D-w_e}\le0.
\]
Since $D-w_e\ge1$ and $w_e>0$, deleting any edge with $b_i<r$ would strictly improve the ratio. Hence $b_i\ge r$. For $w'\ge w$,
\[
\sum_{e\in M}(b_i-r)w'_e\ge\sum_{e\in M}(b_i-r)w_e=r,
\]
hence $R_{b,w'}(M)\ge r$. For empty $M$, nonnegativity suffices.
\end{proof}

The coordinatewise maximum of finitely many admissible tables remains admissible and weakly improves their optimized welfare.

\subsection{Augmenting the Public Range}\label{app:augmentation}
Exact fixed-order optimization with pivots truthfully attains $F$, and enumeration is polynomial for fixed $K$. Its complexity for general contexts is left open here. For $m$ pairwise disjoint edges in editorial order, let $V(j,k)$ be the optimal suffix welfare conditional on reaching edge $j$, with $k$ slots left. The recurrence $V(j,k)=\max\{V(j+1,k),p_jv_j+(1-p_j)V(j+1,k-1)\}$ chooses whether to show edge $j$. With $V(m+1,k)=V(j,0)=0$, backward induction gives $V(1,K)=F$ in $O(mK)$ arithmetic operations, covering the disjoint tight examples.

Thinning admits polynomial-time optimization for arbitrary $K$, and a public list of unthinned renderings can enlarge its range. Fix feasible matchings $M_1,\ldots,M_L$ before bids and render each unthinned in editorial order. In $M_\ell$, edge $e_r=(i,j)$ has public service $s_i^\ell=p_{e_r}\prod_{h<r}(1-p_{e_h})$, with zero service for unmatched advertisers. The augmented optimum is
\[
r^{\mathrm{aug}}(b)=\max\Bigl\{r^{\mathrm{odds}}(b),\max_{1\le\ell\le L}\sum_{i=1}^{n}b_i s_i^\ell\Bigr\}.
\]
Exact optimization and pivots over this bid-independent union preserve truthfulness in expectation and realized individual rationality~\cite{midr}. Welfare weakly increases, retaining every odds approximation guarantee. Each allocation or pivot call adds $O(LK)$ arithmetic operations after preprocessing, so a polynomial-time construction of a polynomial-sized list preserves polynomial complexity. For two disjoint half-probability edges with unit values, adding the full matching raises welfare from $2/3$ to $3/4$. Section~\ref{sec:rangecomparison} evaluates this augmentation with a single $M$ maximizing $\sum_{e\in M}p_e$.

\subsection{Downward-Closed Feasibility Constraints}\label{app:feasibility}
The service construction and welfare comparison extend beyond matching under downward closure. The computational guarantee additionally requires exact optimization over the feasibility family.

\begingroup
\emergencystretch=1em
\begin{corollary}[Feasibility-family reduction]\label{cor:family}
Replace capped matchings by any finite downward-closed family of edge sets containing at most one edge per advertiser. Assume the same cascade with a fixed public order and an exact oracle optimizing the fractional objective over this family. The mechanism is truthful in expectation, realized individually rational, and a factor-two approximation. A polynomial-size integral polyhedral description with polynomial-time exact LP optimization suffices for polynomial parameter computation and recovery of an integral optimizer.
\end{corollary}
\endgroup

The feasibility family sets the oracle's domain, while the service formula and pivot construction remain the same. Downward closure keeps both rendered and comparator subsets feasible.

\begin{proof}
Lemma~\ref{lem:service} applies to each feasible set: its telescoping identity uses only ordered probabilities. One edge per advertiser gives scalar service, so exact range optimization and pivots preserve truthfulness and realized individual rationality. The high-value benchmark subset remains feasible, giving factor two by the proof of Theorem~\ref{thm:main}. For an integral polytope, Lemma~\ref{lem:lfm} gives exact fractional optimization, and Appendix~\ref{app:recovery} recovers an integral optimizer.
\end{proof}

Integrality ensures that the relaxation's optimum is attained by a feasible set, whereas an arbitrary relaxation may overestimate every feasible ratio. Without downward closure, omission can also break feasibility: under an exactly-$K$ requirement, even for $K=1$, basic thinning omits a certain-click singleton with probability $1/2$, producing an infeasible empty rendering.

\section{The Running Example}\label{app:running}

At truthful bids, the basic range has values
\[
\begin{array}{c|c@{\qquad}c|c}
M&R_v(M)&M&R_v(M)\\\hline
\varnothing&0&\{Ax\}&16/9\\
\{Ay\}&2/3&\{Bx\}&9/8\\
\{By\}&27/19&\{Ax,By\}&59/27\\
\{Ay,Bx\}&13/9&&
\end{array}
\]
The diagonal is optimal. In order $Ax,By$, suffix sums $17/10,9/10$ give retention probabilities $10/27,10/19$ and services $q_A=(4/5)/(1+17/10)=8/27$, $q_B=(9/10)/(1+17/10)=1/3$. In particular, $B$ is reached with probability $19/27$ and then clicks with probability $(10/19)(9/10)=9/19$. Welfare is $59/27$, about $58.4\%$ of the freely ordered optimum $187/50$.

Zeroing $A$'s bid gives $H_A=27/19$ at $By$. The selected diagonal gives $B$ reported welfare $3\cdot(1/3)=1$, so $t_A=8/19$. Similarly $H_B=16/9$ and $A$'s selected welfare $32/27$ give $t_B=16/27$. Click charges are $c_A=t_A/q_A=27/19<4$ and $c_B=t_B/q_B=16/9<3$. Only the clicked advertiser pays, with expected payments $t_A$ and $t_B$. Truthful utilities are $u_A=392/513$ and $u_B=11/27$ in expectation.

\begin{table}[t]
\centering
\begin{tabular}{lccc}
\toprule
$A$'s report $x$ & Selected matching & $q_A$ & $u_A$ at $v_A=4$\\
\midrule
$0\le x<27/19$ & $\{By\}$ & $0$ & $0$\\
$27/19<x<27/4$ & $\{Ax,By\}$ & $8/27$ & $392/513$\\
$x>27/4$ & $\{Ax\}$ & $4/9$ & $61/171$\\
\bottomrule
\end{tabular}
\caption{Basic mechanism as $A$'s report $x$ varies ($b_B=3$, $v_A=4$). Breakpoint ties use the public tie rule.}\label{tab:thresholds}
\end{table}

In Table~\ref{tab:thresholds}, each matching's reported welfare is a line in $x$ with slope $q_A$. The upper envelope consists of $27/19$, $8x/27+1$, and $4x/9$, intersecting at $27/19$ and $27/4$. The other four lines lie below it for $x\ge0$. The truthful report lies in the diagonal interval, where utility is constant across reports. In the final interval, utility is $16/9-H_A=61/171$. Thus a sufficiently high report raises service from $8/27$ to $4/9$ but increases the expected payment from $8/19$ to $H_A=27/19$, reducing true utility from $392/513$ to $61/171$.

\section{Experimental Details and Additional Results}\label{app:experiments}

All welfare comparisons use exact rational selection, prices, and cascade welfare, without Monte Carlo sampling of clicks. Values are independent uniform draws from $\{1/4,1,4,16\}$ and the editorial order follows increasing context index. Separable probabilities are $p_{ij}=\delta A_iC_j/100$, with independent uniform $A_i,C_j\in\{1,\ldots,10\}$. Contextual probabilities are $\delta B_{ij}/100$, with independent uniform $B_{ij}\in\{1,\ldots,100\}$. Bottleneck instances use contextual probabilities, allow all edges to the first half of contexts, and allow each remaining edge independently with probability $0.15$. The largest sampled probability may be below $\delta$.

For the 144 small inputs of Section~\ref{sec:rangecomparison}, we enumerate all feasible matchings, using decreasing-value order for free-order OPT and editorial order for fixed-order OPT. Public augmentation adds an unthinned maximum-$\sum p_e$ matching, selected with the public tie rule, and computes pivots over the union. Exact tight-family checks use $K\in\{2,3,4\}$, $\delta\in\{0.1,0.5,0.8\}$, and $\epsilon\in\{\delta/10,\delta/100\}$. In both orientations, odds welfare is one and augmentation attains fixed-order OPT.

For $m=K=4$, Algorithm~2 of the arXiv version of~\cite{anchor} uniformly samples one of four buckets. Buckets $\ell=1,2,3$ cover $(2^{-\ell},2^{-(\ell-1)}]$, and bucket four covers probabilities at most $1/(2m)=1/8$. Within bucket $\ell$, greedy matching selects up to $\min\{2^\ell,m\}$ edges in decreasing $v_ip_{ij}$ order and renders them in that order. We average the four buckets' exact welfare, including empty outputs. The nonempty-bucket variant averages the same conditional allocations over publicly nonempty buckets, independently of bids.

\subsection{Implementation and Welfare Certificates}
Timings use a sequential Python prototype with an exact integer network-simplex solver and numerical dual proposals, measured after warm-up on one desktop CPU. Starting at $r=0$, the rational prototype maximizes $\sum_e(b_i-r)o_e$ and updates $r$ to the selected ratio until the exact residual is zero. Positive residuals strictly increase the ratio, ensuring termination over the finite matching family. Integer scaling preserves the objective, and binary edge penalties implement lexicographic ties. The polynomial guarantee rests on LP~\eqref{eq:lp}.

The separate sparsity suite has 60 contextual instances: $n\in\{16,64,128\}$, $m=8$, $K\in\{4,8\}$, $\delta=0.8$, independent edge density in $\{0.25,1\}$, and five repetitions per combination. Across the six $(n,\text{density})$ groups, minimum certified $W/U$ (with $U$ defined below) ranges from $0.722$ to $0.938$. At $n=128$, median allocation times are $7.9$ ms at density $0.25$ and $22.3$ ms at density one. Allocation uses at most three oracle calls, and at most eight bidders require pivots.

Over the capped matching polytope $P$, we use the last-choice relaxation~\cite{downward}
\[
U_{\mathrm{LP}}=\max\Bigl\{\sum_{e=(i,j)}v_i y_e:\boldsymbol{x}\in P,\ 0\le y_e\le p_ex_e,\ \sum_e y_e\le1\Bigr\}.
\]
Every rendered allocation induces a feasible point, so $\OPT\le U_{\mathrm{LP}}$. Nonnegative rational dual variables satisfying
\[
\alpha_i+\beta_j+\gamma\ge p_e\lambda_e,\qquad\lambda_e+\eta\ge v_i\quad(e=(i,j)\in E)
\]
give $U=\sum_i\alpha_i+\sum_j\beta_j+K\gamma+\eta\ge U_{\mathrm{LP}}$, making $W/U$ conservative. We convert numerical row, column, capacity, and total-click multipliers to nonnegative rationals, set $\lambda_e=(v_i-\eta)_+$, and raise $\gamma$ to the maximum of its proposal and the deficits $p_e\lambda_e-\alpha_i-\beta_j$. All inequalities are checked exactly. The best welfare among odds, its unthinned matching, the best singleton, and standalone-weight matching gives a feasible lower bound.

\subsection{Runtime versus Free-Order DP-VCG}\label{app:paired}
Table~\ref{tab:dp} compares odds with free-order DP-VCG (Appendix~\ref{app:dp}) on 36 dense contextual instances: $n\in\{32,128\}$, $m\in\{4,8,12\}$, and $K\in\{2,4\}$, with three repetitions per combination, two at $\delta=0.05$ and one at $\delta=0.8$. Both methods receive identical inputs, use exact rational arithmetic, and compute pivots only for positive-service winners. Allocation time includes matching recovery, and pivot time covers all pivot computations. The timing order alternates across instances.

\begin{table}[t]

\centering\small
\begin{tabular*}{\columnwidth}{@{\extracolsep{\fill}}ccrrr@{}}
\toprule
Contexts $m$ & Cap $K$ & Odds A/P & DP-VCG A/P & Min. $W/O$\\
\midrule
4 & 2 & 6.1/14.2 & 4.6/8.9 & 0.924\\
4 & 4 & 6.7/30.5 & 9.0/35.8 & 0.914\\
8 & 2 & 13.6/25.3 & 16.6/33.3 & 0.918\\
8 & 4 & 14.6/62.4 & 128.1/518.0 & 0.937\\
12 & 2 & 18.1/43.6 & 41.2/74.4 & 0.924\\
12 & 4 & 24.2/112.1 & 747.2/2830.6 & 0.930\\
\bottomrule
\end{tabular*}
\caption{Median prototype times in ms (A/P: allocation/all pivots), six paired instances per row ($n\in\{32,128\}$). Min. is the smallest odds welfare/free-order OPT in the row.}\label{tab:dp}
\end{table}

Across the 24 low-probability inputs, mean $W/O$ is $0.958$ and its minimum is $0.938$. Every instance satisfies bound~\eqref{eq:sharp}.

\subsection{Welfare across Instance Families}\label{app:structure}
Tables~\ref{tab:ranges} and~\ref{tab:methods} pool eight inputs per structure, totaling 24 per $(\delta,K)$ group. Min. denotes the smallest $W/O$. For $K\in\{2,4\}$, structure-specific mean $W/O$ ranges are $0.892$--$0.981$ for separable, $0.881$--$0.952$ for contextual, and $0.891$--$0.957$ for bottleneck inputs. Corresponding mean $F/O$ ranges are $0.959$--$1$, $0.956$--$0.988$, and $0.986$--$1$. For standalone-weight matching at $\delta=0.8,K=4$, mean welfare ratios are $0.872$, $0.280$, and $0.775$, respectively. For contextual inputs, odds attains $0.952$ versus fixed-order OPT's $0.964$: the heuristic loss far exceeds the order loss.

Across these 96 instances, median $W/F$ is $0.951$, with an interquartile range of $0.909$--$0.982$ (linearly interpolated empirical quartiles). For the single-matching augmentation, the paired gain $(W_{\mathrm{aug}}-W)/F$ has mean $0.0165$, maximum $0.1203$, and median zero.

Retention is total expected displays divided by total selected ads over 48 instances per $\delta$ with $K\in\{2,4\}$. These include 5 and 22 unthinned singleton selections at $\delta=0.1$ and $0.8$, respectively.

\section{Dynamic Program for Free-Order OPT}\label{app:dp}
Sort advertisers by $v_1\ge\cdots\ge v_n$, breaking ties by public identifiers. Let $G(i,U)$ be the optimal welfare from advertisers $i,\ldots,n$ using contexts outside $U$ and $K-|U|$ remaining slots, given that no earlier selected ad was clicked. Independence preserves the remaining click probabilities. With $J_i(U)=\{j\in\{1,\ldots,m\}\setminus U:(i,j)\in E\}$, for $i\le n$ and $|U|<K$ the recurrence is
\begin{equation}\label{eq:dp}
\begin{aligned}
G(i,U)=\max\Bigl\{&G(i+1,U),\\
&\max_{j\in J_i(U)}\bigl[v_ip_{ij}
 +(1-p_{ij})G(i+1,U\cup\{j\})\bigr]\Bigr\}.
\end{aligned}
\end{equation}
The maximum over an empty $J_i(U)$ is omitted. Set $G(n+1,U)=0$ and $G(i,U)=0$ for $|U|=K$. Either advertiser $i$ is omitted, or it occupies some available context $j$ and precedes every later selected advertiser in the value order. In the latter case it contributes $v_ip_{ij}$, and the optimal suffix is reached with probability $1-p_{ij}$. Backward induction on $i$ proves the recurrence and that $G(1,\varnothing)=\OPT(v)$. Following maximizing transitions recovers a matching and its rendering order.

Sorting and evaluating the states take
\[
O\!\left(n\log n+nm\sum_{h=0}^{\min(K,m)}\binom mh\right)
\]
arithmetic operations. This is polynomial for fixed $K$ and fixed-parameter tractable in $m$. The subset-state dependence remains exponential when both $m$ and $K$ grow. Each state is the welfare of one matching of at most $K$ edges, so exact arithmetic adds polynomial input-bit-length cost per operation. Memoization and exact rational comparisons implement the recurrence.

Pivots over the full public range of matchings and orders give the free-order VCG baseline, whose value-based rendering need not preserve editorial order. For each of at most $K$ positive-service winners, zero its bid, re-sort, and rerun the DP. Divide the pivot by its actual click service to obtain the click-contingent charge. Other bidders pay zero.

\section{Robustness to Calibration Error}\label{sec:calibration}
Let $\widehat p$ be the auction's probabilities and $p$ the actual independent, omission-invariant click probabilities, with the same editorial order and allowed edges. Calibration error separates the service used for pricing from actual service.

\begin{proposition}[Bounded calibration error]\label{prop:robust}
Suppose values and allowed reports lie in $[0,V]$, and $\max_e|p_e-\widehat p_e|\le\varepsilon$. Run either the basic or the odds mechanism using $\widehat p$, and set $\Delta=\min\{1,K\varepsilon\}$. Every bidder's expected gain from any unilateral deviation is at most $2V\Delta$. At truthful reports,
\[
W_p\ge\frac{\OPT_p-3V\Delta}{2}.
\]
Truthful realized individual rationality continues to hold.
\end{proposition}
\begin{proof}
At fixed reports, couple the two models using the same retention lottery and independent uniforms for click trials. At most $K$ trials matter. Each coupled trial differs with probability $|p_e-\widehat p_e|$, so a union bound limits the probability of different first clicks to $\Delta$. Welfare in $[0,V]$ changes in expectation by at most $V\Delta$. A bidder's utility is $(v_i-c_i)$ on its click and zero otherwise. Since $0\le c_i\le b_i\le V$, $|v_i-c_i|\le V$ gives the same error bound. Write $U_p(b_i)$ for expected utility at report $b_i$, with other reports fixed. Truthfulness under $\widehat p$ therefore gives
\[
U_p(b_i)-U_p(v_i)\le U_{\widehat p}(b_i)-U_{\widehat p}(v_i)+2V\Delta\le2V\Delta.
\]

The same welfare comparison applies to each freely ordered deterministic allocation, hence $|\OPT_p-\OPT_{\widehat p}|\le V\Delta$. Therefore $W_p\ge W_{\widehat p}-V\Delta\ge\OPT_{\widehat p}/2-V\Delta\ge\OPT_p/2-3V\Delta/2$. At a truthful report the pointwise charge bound $c_i\le v_i$ holds regardless of the actual click probabilities.
\end{proof}

\noindent\textit{Example (Miscalibrated probabilities).} For odds, take two disjoint edges in editorial order $A,B$, with $K=2$, estimates $\widehat p_A=\widehat p_B=1/2$, actual probabilities $p_A=p_B=4/5$, $b_A=1$ and $v_B=3/2$. The truthful report selects both ads, with retention $(2/3,1)$ and click price $c_B=1/2$. $B$'s actual click service is $28/75$, and since $v_B-c_B=1$, its expected utility is also $28/75$. A report of three selects $B$ alone, with retention one, click price one, and actual utility $(4/5)(3/2-1)=2/5$. The deviation gain is $2/5-28/75=2/75$.

\end{document}